\documentclass[12pt,a4paper]{amsart}
\usepackage[centertags]{amsmath}
\usepackage{amsfonts}
\usepackage{amssymb}
\usepackage{amsthm}
\usepackage{newlfont}
\usepackage{a4}
\usepackage{enumerate}

\theoremstyle{definition}
\newtheorem{defn}{Definition}[section]
\newtheorem{thm}[defn]{Theorem}
\newtheorem{lemma}[defn]{Lemma}
\newtheorem{tvr}[defn]{Proposition}
\newtheorem{cor}[defn]{Corollary}
\theoremstyle{remark}
\newtheorem{example}{Example}[section]

\newcommand{\der}{\operatorname{der}}
\newcommand{\Span}{\operatorname{span}}
\newcommand{\inv}{\operatorname{inv}}
\newcommand{\codim}{\operatorname{codim}}
\newcommand{\mL}{\mathcal{L}}
\newcommand{\mD}{\mathcal{D}}
\newcommand{\CC}{\mathbb{C}}
\newcommand{\End}{\operatorname{End}}
\newcommand{\gl}{\operatorname{gl}}
\newcommand{\jor}{\operatorname{jor}}
\newcommand{\f}{\operatorname{\psi}}
\newcommand{\g}{\operatorname{\phi}}

\begin{document}

\title[On $(\alpha,\beta,\gamma)$--derivations of Lie
Algebras and Invariant Function]{On
$(\alpha,\beta,\gamma)$--derivations of Lie Algebras and
Corresponding Invariant Functions}

\author{Petr Novotn\'y, Ji\v{r}\'\i\ Hrivn\'ak}%

\begin{center}\Large\bf
On
$(\alpha,\beta,\gamma)$--derivations of Lie Algebras and
Corresponding Invariant Functions
\end{center}
\bigskip
\begin{center}\large
Petr Novotn\'y, Ji\v{r}\'\i\ Hrivn\'ak\footnote{Corresponding author: Tel.: +420 2 24358351; fax: +420 2 22320861 \newline {\it E-mail address:} jiri.hrivnak@fjfi.cvut.cz (J. Hrivn\'ak)}
\end{center}
\bigskip
\begin{center}
Department of Physics,
Faculty of Nuclear sciences and Physical Engineering, Czech
Technical University, B\v{r}ehov\'a 7, 115 19 Prague 1, Czech
Republic
\end{center}
\vspace{24pt}
\hrule\vspace{8pt}
{\footnotesize
ABSTRACT. We consider finite--dimensional complex Lie algebras. We generalize
the concept of Lie derivations via certain complex parameters and
obtain various Lie and Jordan operator algebras as well as two
one--parametric sets of linear operators. Using these parametric
sets, we introduce complex functions with fundamental property --
invariance under Lie isomorphisms. One of these basis--independent
functions represents a complete set of invariant(s) for
three--dimensional Lie algebras. We present also its application on
physically motivated examples in dimension eight.
}\vspace{8pt}
\hrule
\vspace{24pt}

\section{Introduction}
The theory of finite dimensional complex Lie algebras is an
important part of Lie theory. It has several applications to physics
and connections to other parts of mathematics. With an increasing
amount of theory and applications concerning Lie algebras of various
dimensions, it is becoming necessary to ascertain applicable tools
for handling them. The miscellaneous characteristics of Lie algebras
constitute such tools and have also found applications: {\it Casimir
operators} \cite{inv}, {\it derived}, {\it lower central} and {\it
upper central sequences}, {\it Lie algebra of derivations}, {\it
radical}, {\it nilradical}, {\it ideals}, {\it subalgebras}
\cite{Jacobson,ide} and recently {\it megaideals} \cite{Pop1}. These
characteristics are particularly crucial when considering possible
affinities among Lie algebras.

Physically motivated relations between two Lie algebras, namely {\it
contractions} and {\it deformations}, have been extensively studied,
see e.g. \cite{G,PdeM}. When investigating these kinds of relations
in dimensions higher than five, one can encounter insurmountable
difficulties. Firstly, aside the {\it semisimple} ones, Lie algebras
are completely classified only up to dimension 5 and the {\it
nilpotent} ones up to dimension 6. In higher dimensions, only
special types, such as {\it rigid} Lie algebras \cite{Goze} or Lie
algebras with fixed structure of nilradical, are only classified
\cite{SW} (for detailed survey of classification results in lower
dimensions see e.g. \cite{Pop1} and references therein). Secondly,
if all available characteristics of two results of
contraction/deformation are the same then one cannot distinguish
them at all. This often occurs when the result of a contraction is
one--parametric or more--parametric class of Lie algebras.

The aim of this article is to partially overcome these obstacles and
to add new objects to the existing set of invariants. Number of
Casimir operators, dimensions of radical, nilradical, lower central
sequences, etc., are all invariant or --- equivalently --- basis
independent. However, in this article we pursue a different kind of
basis independent characteristics --- certain complex functions.
These invariant functions, which arise from the concept of so called
$(\alpha,\beta,\gamma)$--derivations, then represent a very powerful
tool for the description of Lie algebras, very effective and
essential when dealing with their parametric continuum.

\medskip
In Section \ref{gender}, we generalize the concept of
derivation of a Lie algebra; we introduce
$(\alpha,\beta,\gamma)$--derivations and show their pertinent
properties. All possible intersections of spaces containing these
derivations are investigated. Examples for low--dimensional Lie
algebras are presented.

In Section \ref{Invar}, we introduce two invariant functions
corresponding to $(\alpha,\beta,\gamma)$--derivations. We
demonstrate on all three--dimensional complex Lie algebras and on
physically motivated examples in dimension eight, how these
functions effectively enlarge the set of 'classical' invariants.

In Section 4, we shortly review other
generalizations of derivations, make a note on a computation of
$(\alpha,\beta,\gamma)$--derivations and further comments.


\section{$(\alpha,\beta,\gamma)$--derivations}\label{gender}

In this article let $\mL$ denote the finite--dimensional Lie algebra
over the field of complex numbers $\CC$ and $\End(\mL)$ the
associative algebra of all linear operators on the vector space
$\mL$. The space $\End(\mL)$, endowed with standard Lie commutator
$[A,B]=AB-BA$, is denoted as usual by $\gl (\mL)$ and the space
$\End(\mL)$, endowed with Jordan product $A\circ
B=\frac{1}{2}(AB+BA)$, is denoted by $\jor (\mL)$. In this way, any
subalgebra of $\End(\mL)$ forms also a subalgebra of $\gl (\mL)$ and
$\jor (\mL)$. We also adopt notation for the {\it center} $C(\mL)$
and for the {\it derived algebra} $\mL^2 = [\mL,\mL]$.

\subsection{Properties and structure of $(\alpha,\beta,\gamma)$--derivations
}\label{genderDP}

Recall that a {\it derivation} of $\mL$ is a linear operator $A
\in \End (\mL)$ such that for all $x,y\in \mL$
\begin{equation}\label{der}
   A[x,y] = [A x,y] + [x,A y]
\end{equation}
and the set of all these derivations, denoted by $\der(\mL)$, forms
a Lie algebra of derivations. Several non--equivalent ways
generalizing this definition have been recently studied
\cite{Bresar,Leger,Hartwig}. However, we will bring forward another
type of generalization.

We call a linear operator $A\in \End(\mL)$ an
\textbf{$(\alpha,\beta,\gamma)$--derivation} of $\mL$ if there exist
$\alpha,\beta,\gamma\in \CC$ such that for all $x,y\in \mL$
\begin{equation}\label{gd}
\alpha A [x,y] = \beta [Ax,y] + \gamma [x,Ay].
\end{equation}
For given $\alpha, \beta, \gamma \in \CC$ we denote the set of all
$(\alpha,\beta,\gamma)$--derivations as
$\mD(\alpha,\beta,\gamma)$, i.e.
\begin{equation}
\mD(\alpha,\beta,\gamma) = \{ A \in \End(\mL)\ |\ \alpha A [x,y] =
\beta [Ax,y] + \gamma [x,Ay],\ \ \forall x,y \in \mL\}.
\end{equation}
Let us focus on this set and show some its properties. It is clear
that $\mD(\alpha,\beta,\gamma)$ is a linear subspace of $\End(\mL)$
and it follows immediately from (\ref{gd}) that for any $\varepsilon
\in \CC \backslash\{0\}$ it holds
\begin{equation}\label{vla1}
\mD(\alpha,\beta,\gamma) =
\mD(\alpha\varepsilon,\beta\varepsilon,\gamma\varepsilon) =
\mD(\alpha,\gamma,\beta).
\end{equation}
Furthermore, we have the following important property:
\begin{lemma} For any $\alpha, \beta, \gamma \in \CC$
\begin{equation}\label{vla2}
\mD(\alpha,\beta,\gamma) = \mD(0,\beta-\gamma,\gamma -\beta) \cap
\mD(2\alpha,\beta+\gamma,\beta+\gamma).
\end{equation}
\end{lemma}

\begin{proof} Suppose any $\alpha, \beta, \gamma \in \CC$ are given. Then for $A \in \mD(\alpha,\beta,\gamma)$
 and arbitrary $x,y\in \mL$ we have
    \begin{eqnarray}
    \alpha A[x,y] & = & \beta[Ax,y] + \gamma [x,Ay] \label{trik1} \\
    \alpha A[y,x] & = & \beta[Ay,x] + \gamma [y,Ax]. \nonumber
    \end{eqnarray}
    By summing and subtracting equations (\ref{trik1}) we obtain
     \begin{eqnarray}
    0 & = & (\beta -\gamma)([Ax,y] - [x,Ay]) \label{trik2} \\
    2\alpha A[x,y] & = & (\beta + \gamma)([Ax,y] + [x,Ay]) \nonumber
     \end{eqnarray}
    and thus $\mD(\alpha,\beta,\gamma) \subset \mD(0,\beta-\gamma,\gamma -\beta) \cap
\mD(2\alpha,\beta+\gamma,\beta+\gamma)$. Similarly, starting with
equations (\ref{trik2}) we obtain equations (\ref{trik1}) and the
remaining inclusion is proven.
\end{proof}
Further, we proceed to formulate the theorem which reveals the
structure of the spaces $\mD(\alpha,\beta,\gamma)$; the three
original parameters are in fact reduced to only one.
\begin{thm}\label{klass}
For any $\alpha, \beta, \gamma \in \CC$ there exists $\delta\in
\CC$ such that the subspace $\mD(\alpha,\beta,\gamma) \subset \End
(\mL)$ is equal to some of the four following subspaces:
\begin{enumerate}
 \item $\mD(\delta,0,0)$
 \item $\mD(\delta,1,-1)$
 \item $\mD(\delta,1,0)$
 \item $\mD(\delta,1,1)$.
\end{enumerate}
\end{thm}
\begin{proof}
\begin{enumerate}
    \item Suppose $\beta + \gamma = 0 $. Then either $\beta = \gamma
    =0$ or $\beta = -\gamma \neq 0$.
    \begin{enumerate}
    \item For $\beta = \gamma = 0$ we
    have $$\mD(\alpha,\beta,\gamma)  =  \mD(\alpha,0,0).$$
    \item For $\beta = -\gamma \neq 0$ we have according to (\ref{vla1}),
    (\ref{vla2}):
 $$\mD(\alpha,\beta,\gamma) = \mD(0,\beta-\gamma,\gamma-\beta)\cap \mD(2\alpha,0,0)
    = \mD(0,1,-1)\cap \mD(\alpha,0,0).$$ On the other hand it holds, $$ \mD(\alpha,1,-1) = \mD(0,2,-2)\cap \mD(2\alpha,0,0)
    = \mD(0,1,-1)\cap \mD(\alpha,0,0)$$ and therefore
    $$\mD(\alpha,\beta,\gamma)= \mD(\alpha,1,-1).$$
    \end{enumerate}
    \item Suppose $\beta + \gamma \neq 0 $. Then either $\beta - \gamma \neq
    0$ or $\beta = \gamma \neq 0 $.
    \begin{enumerate}
        \item For $\beta - \gamma \neq 0$ we have
    $$\mD(\alpha,\beta,\gamma) = \mD(0,\beta-\gamma,\gamma-\beta)\cap
    \mD(2\alpha,\beta+\gamma,\beta+\gamma) = \mD(0,1,-1)\cap
    \mD(\frac{2\alpha}{\beta+\gamma},1,1)$$
    and this is according to (\ref{vla2}) equal to
    $\mD(\frac{\alpha}{\beta+\gamma},1,0)$, i.e.
    $$\mD(\alpha,\beta,\gamma) =
    \mD(\frac{\alpha}{\beta+\gamma},1,0).$$
        \item For $\beta = \gamma \neq 0 $ we have
        $$\mD(\alpha,\beta,\gamma) = \mD(\frac{\alpha}{\beta},1,1).$$
    \end{enumerate}
\end{enumerate}
\end{proof}

Now we will discuss in detail the possible outcome of Theorem
\ref{klass} which depends on the value of the parameter $\delta\in
\CC$.
    \begin{enumerate}
        \item $\mD(\delta,0,0)$:
                \begin{enumerate}
                \item For $\delta = 0 $ we trivially get $\mD(0,0,0) =
                \End(\mL)$.
                \item For $\delta \neq 0$ the space $\mD(1,0,0)$ is an associative subalgebra of $\End (\mL)$, which maps derived
                algebra $\mL^2 = [\mL,\mL]$ to zero vector:
                $$\mD(1,0,0) = \{A \in \End(\mL)\ |\ A(\mL^2) =
                0\},$$
                and therefore its dimension is as follows:
                       $$\dim \mD(1,0,0) = \operatorname{codim}\mL^2 \dim
                       \mL.$$
                \end{enumerate}

        \item $\mD(\delta,1,-1)$:
                \begin{enumerate}
                \item For $\delta = 0$ we have a Jordan algebra
                $\mD(0,1,-1)\subset \jor (\mL)$,
                $$\mD(0,1,-1) = \{A \in \End(\mL)\ |\ [Ax,y] = [x,A y],\ \ \forall
                x,y\in\mL \}.$$
                \item For $\delta \neq 0$ we get Jordan algebra $\mD(1,1,-1)\subset \jor (\mL)$ as an intersection of two Jordan
                algebras:
                $$ \mD(\delta,1,-1) = \mD(0,1,-1) \cap
                \mD(\delta,0,0) = \mD(0,1,-1) \cap \mD(1,0,0)=\mD(1,1,-1) .$$
                \end{enumerate}

        \item $\mD(\delta,1,0)$:

                \begin{enumerate}
                \item For $\delta = 0$ we get an associative algebra of all linear operators
                of the vector space $\mL$, which maps the whole $\mL$ into its center $ C(\mL)$:
                $$ \mD(0,1,0) = \{A \in \End(\mL)\ |\ A(\mL) \subseteq
                C(\mL)\}, $$ and its dimension is
                        $$\dim \mD(0,1,0) = \dim \mL \, \dim C(\mL).$$

                \item For $\delta = 1$ the space $\mD(1,1,0)$ is the centralizer of adjoint representation ad($\mL$) in $\gl (\mL)$.

                \item For the remaining values of $\delta$ the space $\mD(\delta,1,0)$
                forms, in the general case of Lie algebra $\mL$,
                only the vector subspace of
                $\End(\mL)$. Thus, we have the one--parametric set of vector
                spaces:
                 $$ \mD(\delta,1,0) = \mD(0,1,-1) \cap
                \mD(2\delta,1,1). $$
                \end{enumerate}

        \item $\mD(\delta,1,1)$:
                \begin{enumerate}
                \item For $\delta = 0$ we have a Lie algebra
                $$\mD(0,1,1) = \{A \in \End(\mL)\ |\ [Ax,y] = -[x,A y],\ \ \forall x,y\in\mL \}.$$
                \item For $\delta = 1$ we get the algebra of derivations of
                $\mL$,
                    $$ \mD(1,1,1) = \der (\mL).$$
                \item For the remaining values of $\delta$ the space $\mD(\delta,1,1)$
                forms, in the general case of Lie algebra $\mL$,
                 only the vector subspace of
                $\End(\mL)$.
                \end{enumerate}
    \end{enumerate}

\subsection{Intersections of the spaces $\mD(\alpha,\beta,\gamma)$
}\label{intersections}

Various intersections of two different subspaces
$\mD(\alpha,\beta,\gamma)$ also turned out to be of interest;
therefore we systematically explore all possible intersections of
these spaces. However, all intersections of these spaces lead us
to only two new structures. The first one
\begin{equation}\label{inter1}
\begin{array}[t]{lll}
    \mD(1,0,0) \cap \mD(0,1,0) & = & \mD(1,0,0) \cap \mD(\delta,1,0) = \mD(1,1,-1) \cap
    \mD(\delta,1,0)\\
           & = & \mD(1,1,-1) \cap \mD(\delta,1,1) = \mD(\delta,1,0) \cap \mD(\gamma,1,0)_{\delta \neq
           \gamma}\\
           & = & \mD(\delta,1,0) \cap \mD(\gamma,1,1)_{2\delta \neq
           \gamma}\\
    \end{array}
\end{equation}
forms an associative algebra and is contained in all spaces
$\mD(\alpha,\beta,\gamma)$. Its dimension is

\begin{equation}\label{inter2}
\dim(\mD(1,0,0)\cap\mD(0,1,0)) = \operatorname{codim} \mL^2
    \dim C(\mL).
\end{equation}

The second one is a new Lie algebra:
\begin{equation}\label{prun} \mD(1,0,0) \cap \mD(0,1,1) =
\mD(1,0,0) \cap \mD(\delta,1,1) = \mD(\delta,1,1) \cap
\mD(\gamma,1,1)_{\delta \neq
    \gamma}.\end{equation}
Other intersections lead to structures, which we already have:
    $$\mD(1,1,-1) = \mD(1,0,0) \cap \mD(0,1,-1)=\mD(1,0,0)\cap \mD(1,1,-1)=\mD(0,1,-1)\cap \mD(1,1,-1)$$
\begin{equation}\label{inter3}
\mD(\delta,1,0) = \mD(0,1,-1) \cap \mD(\delta,1,0) = \mD(0,1,-1) \cap
    \mD(2\delta,1,1)=\mD(\delta,1,0)\cap \mD(2\delta,1,1).
\end{equation}
For completeness we state that the space $\mD(1,0,0) \cap
\mD(0,1,0)$ forms an ideal in $\mD(1,1,0),\ \mD(1,1,1)$ and in $
\mD(1,1,-1)$; the space $\mD(1,0,0) \cap \mD(0,1,1) $ is an ideal
in $\mD(1,1,1)$; the space $\mD(0,1,0)$ forms an ideal in
$\mD(0,1,1)$ and in $ \mD(0,1,-1)$.

The structure of algebras $\mD(1,0,0)$, $\mD(0,1,0)$ and their
intersection depends only on the dimensions of the center and the
derived algebra of $\mL$. If $\mL$ and $\widetilde{\mL}$ are the Lie
algebras of the same dimension then the same dimensions of their
centers imply $\mD(0,1,0) \cong \widetilde{\mD}(0,1,0)$; the
coinciding dimension of the derived algebra implies $\mD(1,0,0)
\cong \widetilde{\mD}(1,0,0).$ Moreover, if $\mL$ and
$\widetilde{\mL}$ are both indecomposable and dimensions of the
centers coincide, as well as the dimensions of derived algebras,
then it holds:
\begin{equation}\label{prunnn}
\mD(0,1,0) \cap \mD(1,0,0) \cong \widetilde{\mD}(0,1,0) \cap
\widetilde{\mD}(1,0,0) .\end{equation}

\subsection{Examples of $(\alpha,\beta,\gamma)$--derivations
}\label{exder} We present as illustration some examples of
$(\alpha,\beta,\gamma)$--derivations for lower--dimensional Lie
algebras. Note especially the form of the one--parametric subspace
$\mD(\delta,1,1)$, as it is significant later on.

\begin{example}

    Consider a two--dimensional Lie algebra $\mL_2$ with a basis
    $\{e_1,e_2\}$ and its only non--zero commutation relation: $[e_1,e_2]=e_2$.
    \begin{enumerate}[-] \itemsep 6pt
        \item $\mD(1,1,1) = \Span_{\CC}
 {\left\{
\left(%
\begin{smallmatrix}
  0 & 0 \\
  1 & 0 \\
\end{smallmatrix}%
\right),
\left(%
\begin{smallmatrix}
  0 & 0 \\
  0 & 1 \\
\end{smallmatrix}%
\right) \right\}} \cong \mL_2$


    \item $\mD(0,1,1) = \Span_{\CC}
{\left\{
\left(%
\begin{smallmatrix}
  0 & 1 \\
  0 & 0 \\
\end{smallmatrix}%
\right),
\left(%
\begin{smallmatrix}
  0 & 0 \\
  1 & 0 \\
\end{smallmatrix}%
\right),
\left(%
\begin{smallmatrix}
  1 & 0 \\
  0 & -1 \\
\end{smallmatrix}%
\right) \right\}} \cong sl(2,\CC)$


    \item $\mD(1,1,0) = \mD(0,1,-1) = \Span_{\CC}{\left\{
\left(%
\begin{smallmatrix}
  1 & 0 \\
  0 & 1 \\
\end{smallmatrix}%
\right) \right\}}$

       \item $\mD(1,0,0) \cap \mD(0,1,1) = \Span_{\CC}{\left\{
\left(%
\begin{smallmatrix}
  0 & 0 \\
  1 & 0 \\
\end{smallmatrix}%
\right) \right\}}$

\item $\mD(\delta,1,0) = \{0\}$ for $\delta \neq 1$.

\item $\mD(\delta,1,1) = \Span_{\CC}
 {\left\{\left(%
\begin{smallmatrix}
  0 & 0 \\
  1 & 0 \\
\end{smallmatrix}%
\right),
\left(%
\begin{smallmatrix}
  \delta -1& 0 \\
  0 & 1 \\
\end{smallmatrix}%
\right) \right\}} $ for $\delta \neq 0$.

    \item $\mD(0,1,0)=\mD(1,0,0) \cap \mD(0,1,0) = \{0\}$

    \item $\mD(1,0,0) = \Span_{\CC}
\left\{
\left(%
\begin{smallmatrix}
  1 & 0 \\
  0 & 0 \\
\end{smallmatrix}%
\right),
\left(%
\begin{smallmatrix}
  0 & 0 \\
  1 & 0 \\
\end{smallmatrix}%
\right) \right\} \cong \mL_2$


    \item $\mD(1,1,-1) = \{0\}$

    \end{enumerate}
\end{example}

\begin{example} Consider a simple Lie algebra of the lowest dimension:
$sl(2,\CC)$.
        \begin{enumerate}[-] \itemsep 6pt
            \item $\mD(1,1,1) \cong sl(2,\CC)$
            \item $\mD(0,1,1) = \mD(0,1,0) = \mD(1,0,0) = \mD(1,1,-1) = \{0\}$
            \item $\mD(1,1,0) = \mD(0,1,-1) = \Span_{\CC}{\left\{\left(\begin{smallmatrix}
              1 & 0 & 0 \\
              0 & 1 & 0 \\
              0 & 0 & 1 \\
              \end{smallmatrix}\right)\right\}}$
            \item $\mD(\delta,1,0) = \{0\}$ for $\delta \neq 1$.
            \item $\mD(\delta,1,1) = \{0\}$ for $\delta \neq \pm
            1, 2$
            \item $\mD(2,1,1) = \Span_{\CC}{\left\{\left(\begin{smallmatrix}
              1 & 0 & 0 \\
              0 & 1 & 0 \\
              0 & 0 & 1 \\
              \end{smallmatrix}\right)\right\}}$
            \item $\mD(-1,1,1) =
            \Span_{\CC} {\left\{\left(
            \begin{smallmatrix}
              1 & 0 & 0 \\
              0 & -2 & 0 \\
              0 & 0 & 1 \\
              \end{smallmatrix} \right),
            \left( \begin{smallmatrix}
              0 & 1 & 0 \\
              0 & 0 & 2 \\
              0 & 0 & 0 \\
              \end{smallmatrix} \right),
              \left( \begin{smallmatrix}
              0 & 0 & 1 \\
              0 & 0 & 0 \\
              0 & 0 & 0 \\
              \end{smallmatrix} \right),
              \left( \begin{smallmatrix}
              0 & 0 & 0 \\
              2 & 0 & 0 \\
              0 & 1 & 0 \\
              \end{smallmatrix} \right),
              \left( \begin{smallmatrix}
              0 & 0 & 0 \\
              0 & 0 & 0 \\
              1 & 0 & 0 \\
              \end{smallmatrix} \right)
              \right\}} $
  \end{enumerate}
\end{example}


\section{Invariant functions}\label{Invar}

\subsection{Definition and properties of functions $\f$ and $\g$}

In this section we define complex functions with pertinent property
--- invariance under Lie isomorphisms. Suppose we have an arbitrary
non--singular linear mapping $\sigma$ and this mapping represents an
isomorphism between two Lie algebras, say $\mL$ and $
\widetilde{\mL}$. That means that for $\sigma:\mL \rightarrow
\widetilde{\mL}$ and for all $x,y\in \widetilde{\mL}$
$$[x,y]_{ \widetilde{\mL}}=\sigma [\sigma^{-1}x,\sigma^{-1}y
]_{\mL}. $$ By rewriting definition relation (\ref{gd}) we have for
$A\in \mD(\alpha,\beta,\gamma)$  $$\alpha A
[\sigma^{-1}x,\sigma^{-1}y]_{\mL} = \beta
[A\sigma^{-1}x,\sigma^{-1}y]_{\mL} + \gamma
[\sigma^{-1}x,A\sigma^{-1}y]_{\mL} .$$ Applying the mapping $\sigma$
on this equation and taking into account that $\alpha, \beta,
\gamma$ are in $ \CC$ then
\begin{equation}\label{invar}
\alpha \sigma A \sigma^{-1} [x,y]_{ \widetilde{\mL}} = \beta
[\sigma A \sigma^{-1}x,y]_{ \widetilde{\mL}} + \gamma [x,\sigma A
\sigma^{-1}y]_{ \widetilde{\mL}},
\end{equation}
i.e. $\sigma A\sigma^{-1}\in \widetilde{\mD}(\alpha,\beta,\gamma)$.
Thus, we easily arrive to the crucial result which we sum up as
follows:

\begin{tvr}\label{tvr1}
Let $\sigma:\mL \rightarrow \widetilde{\mL}$ be an isomorphism of
Lie algebras $\mL$ and $\widetilde{\mL}$. Then the mapping $\rho
:\End(\mL)\rightarrow \End(\widetilde{\mL})$ defined for all $A\in
\End(\mL)$ by $ \rho (A) = \sigma A \sigma^{-1} $ is an
isomorphism of associative algebras $\End(\mL)$ and $
\End(\widetilde{\mL})$. Moreover, for any $\alpha, \beta,
\gamma\in \CC$
 $$\rho(\mD(\alpha,\beta,\gamma)) =
\widetilde{\mD}(\alpha,\beta,\gamma)$$ holds.
\end{tvr}
\begin{cor}
For any $\alpha, \beta, \gamma\in \CC$ the dimension of the vector
space $\mD(\alpha,\beta,\gamma)$ is an invariant of the Lie algebra
$\mL$.
\end{cor}

Indeed, the possibility of new invariants as dimensions of various
$\mD(\alpha,\beta,\gamma)$ is very promising; it seems timely to
list a summary of those spaces $\mD(\alpha,\beta,\gamma)$ whose
dimensions do not depend on dimensionality of well--known
substructures of $\mL$, such as center $C (\mL )$ and $\mL ^2$. The
outcome of Theorem \ref{klass} and the discussion below it,
relations (\ref{prun}), (\ref{prunnn}), yield:
\begin{enumerate}
    \item associative algebra $\mD(1,1,0) $ \label{ggg}
    \item Lie algebras $\mD(1,1,1),\ \mD(0,1,1),\
    \mD(1,0,0) \cap \mD(0,1,1)$
    \item Jordan algebras $\mD(1,1,-1),\ \mD(0,1,-1)$
    \item one--parametric sets of vector spaces $\mD(\alpha,1,0), \
    \mD(\alpha,1,1)$.
\end{enumerate}
Since the definition of $(\alpha,\beta,\gamma)$--derivations
partially overlaps other generalizations, some of these sets
naturally appeared already in the literature. Namely in
\cite{Leger}, a considerable amount of theory concerning relations
between $\mD(1,1,0) $ and $\mD(0,1,-1)$ has been developed (see also
Concluding remarks). In \cite{NH}, the usefulness of the invariant
dimensions of the Lie structures $\mD(1,1,1)$, $\mD(0,1,1)$,
$\mD(1,1,0) $ and $\mD(1,0,0) \cap \mD(0,1,1)$ as well as their
mutual independence has been shown. In this article we focus on
one--parametric sets of vector spaces $\mD(\alpha,1,0)$ and
$\mD(\alpha,1,1)$.

We use these one--parametric sets of vector spaces to define the
invariant function of a Lie algebra $\mL$. Functions $\f,\g:\CC
\longrightarrow \{0,1,2,\ldots,(\dim\mL)^2\}$ defined by formulas
\begin{equation}
\f(\alpha) = \dim\mD(\alpha,1,1)
\end{equation}
\begin{equation}
\g(\alpha) = \dim\mD(\alpha,1,0)
\end{equation}
are called {\bf invariant functions} corresponding to
$(\alpha,\beta,\gamma)$--derivations of a Lie algebra $\mL$. We
observe that from the relations (\ref{inter1}),(\ref{inter2}) and
(\ref{inter3}) it follows: $$\codim \mL^2 \dim C( \mL )\leq \g
(\alpha) \leq \dim \mD (0,1,-1) $$ $$\g (\alpha)\leq \f (2\alpha) $$
$$\codim \mL^2 \dim C( \mL )\leq \f(\alpha) . $$

From proposition \ref{tvr1} it follows immediately that for two Lie
algebras $\mL$ and $\widetilde{\mL}$ it holds:  $$\mL \cong
\widetilde{\mL} \Rightarrow \f = \widetilde{\f}\ \operatorname{and}
\ \g = \widetilde{\g}.$$

Note that sometimes in the literature, the name ''invariant
function'' denotes a (formal) Casimir invariant; its form however
depends on the choice of a basis of $\mL$. Here by invariant
functions we rather mean 'basis independent' complex functions, such
as $\f$ and $\g$. Since the purpose of the functions $\f$ a $\g$ is
to enlarge the set of 'classical' invariants which we list bellow,
this terminology is well justified.

Classical method of identification of an (indecomposable) Lie
algebra \cite{ide} boils down to computation of derived series
$D^{k+1}(\mL) = [D^k(\mL),D^k(\mL)],\ D^0(\mL)=\mL$, lower central
series $\mL^{k+1} = [\mL^k,\mL],\ \mL^1 = \mL$, and upper central
series $$C^{k+1}(\mL)/C^k(\mL)=C(\mL/C^k(\mL)),\, C^1 (\mL)=C (\mL).$$ The dimensions of
these ideals and the dimensions of the spaces of $(\alpha,\beta,\gamma)$--derivations are ''numerical'' invariants of the Lie algebra as well
as the
number of formal Casimir invariants $\tau(\mL)$ \cite{inv}. 
These characteristics, applied on radical, nilradical and factors of
{\it Levi decomposition} of $\mL$, also form invariants. We adopt
the following notation $$
\begin{array}{lll}

    \inv(\mL) & = & (\dim D^k(\mL) )\ (\dim \mL^k )\ (\dim C^k(\mL) )\
    \tau(\mL) \\
    && [\dim \mD(1,1,1) , \dim \mD(0,1,1) , \dim \mD(1,1,0) ,\\
    &&  \dim(\mD(1,1,1) \cap \mD(0,1,1)), \dim \mD(1,1,-1),\ \dim
    \mD(0,1,-1)].
\end{array}
$$

\subsection{Application of the invariant function $\f$ to three--dimensional Lie algebras}
Since among three--dimensional Lie algebras infinite continuum
already appears, it is clear that the finite set $\inv(\mL)$ of
certain dimensions, though useful, can never completely characterize
Lie algebras of dimension higher or equal to 3. On the contrary, it
turns out that the invariant function $\f$ alone(!) forms a complete
set of invariant(s) for indecomposable 3--dimensional Lie algebras,
as shows Table 1. We use the notation for 3--dimensional Lie
algebras as in \cite{ccc}. We point out the case of $A_{3,5}(a)$,
where the function $\f$ is different for different values of the
parameter $a\in \CC$, $0<|a|<1$ or $|a|=1,$ Im$\,a>0$ and thus
distinguishes among non--isomorphic algebras in this continuum. One
may also find convenient that once having some structure constants
of any indecomposable 3--dimensional Lie algebra, simple computation
of function $\f$ allows an unambiguous identification in the list
(see also Concluding remarks). We may also mention here that the
invariant function $\g$ has for $A_{3,1}$ a single value $\g(\alpha)
= 3$, and for the remaining algebras $A_{3,i},\,i=2,\dots , 8$ it
holds:
$$\g(\alpha) = \left\{
\begin{matrix}
    1, & \alpha = 1 \\
    0, & \alpha \neq 1 .
  \end{matrix} \right.
$$


\begin{table}[h!]
\vspace{-1pt} {\renewcommand{\arraystretch}{1.2}
\begin{tabular}[t]{|l|l|l|l||l|l|l|l|}
\hline
\parbox[l][40pt][c]{0pt}{} \textbf{$\mL$} & \textbf{Commutators} & \textbf{inv($\mL$)} & \multicolumn{5}{l}{\textbf{Function $\f$}}\vline\\

\hline\hline $A_{3,1}$ & $[e_2,e_3]= e_1$ & (310)(310)(13)\ 1 &    $\alpha$ & \multicolumn{4}{l}{} \vline\\
\cline{4-8}                &                  & [6, 6, 3, 5, 3, 4]&  $\f(\alpha)$ & 6 & \multicolumn{3}{l}{} \vline\\

\hline\hline $A_{3,2}$ & $[e_1,e_3] = e_1,$    &  (320)(32)(0)\ 1 &  $\alpha$ & 1 & \multicolumn{3}{l}{} \vline\\
\cline{4-8}           & $[e_2,e_3] = e_1+e_2$ &  [4, 3, 1, 2, 0, 1] & $\f(\alpha)$ & 4 & 3 & \multicolumn{2}{l}{} \vline\\

\hline\hline $A_{3,3}$ & $[e_1,e_3]= e_1,$ &   (320)(32)(0)\ 1 &  $\alpha$ & 1 & \multicolumn{3}{l}{} \vline \\
\cline{4-8}           & $[e_2,e_3]= e_2$  &  [6, 3, 1, 2, 0, 1] & $\f(\alpha)$ & 6 & 3 & \multicolumn{2}{l}{} \vline\\

\hline\hline $A_{3,4}$ & $[e_1,e_3]= e_1,$ &  (320)(32)(0)\ 1 &  $\alpha$ & -1 & 1 & \multicolumn{2}{l}{} \vline\\
\cline{4-8}           & $[e_2,e_3]= -e_2$ &[4, 3, 1, 2, 0, 1] &  $\f(\alpha)$ & 5  & 4 & 3 & \\

\hline\hline $A_{3,5}(a)$ & $[e_1,e_3]= e_1$,\  &  (320)(32)(0)\ 1 &  $\alpha$ & 1 & $a$ & $\frac{1}{a}$ & \\
\cline{4-8}              & $[e_2,e_3]= ae_2$,\ $\begin{smallmatrix} 0<|a|<1\\  \text{or}\  |a|=1,\,\text{Im}\,a>0 \\\end{smallmatrix} $ & [4, 3, 1, 2, 0, 1] & $\f(\alpha)$ & 4 & 4 & 4 & 3\\

\hline\hline $A_{3,8}$ & $[e_1,e_3]= -2e_2,$ & (3)(3)(0)\ 1 & $\alpha$ & -1 & 1 & 2 &\\
\cline{4-8}           & $[e_1,e_2]= e_1,\ [e_2,e_3] = e_3$ & [3,
0,
1, 0, 0, 1] & $\f(\alpha)$ & 5  & 3 & 1 & 0\\
\hline
\end{tabular}
}
\medskip

\centering \caption[l]{\it Indecomposable three--dimensional complex
Lie algebras and their invariant function $\f$. Blank space in
smaller table of function $\f$ denotes general complex number,
different from all previously listed values, e.g. for $A_{3,8}$ it
holds: $\f(\alpha)=0$, $\alpha\neq -1,1,2$. }
\vspace{-8pt}\label{tab1}
\end{table}


\subsection{Application of the function $\f$ to some eight--dimensional Lie
algebras}\label{higher}

\begin{example}Let us introduce two eight--dimensional complex Lie algebras
$\mL$, $\widetilde{\mL}$ by listing their commutation relations in
the basis $\{ e_1,\dots, e_8\}$:  $$
\begin{array}{ll}
\mL \quad & [e_1,e_3]=e_5,\ [e_1,e_4]=e_8,\ [e_1,e_5]=e_7,\
[e_1,e_6]=e_4,\\
            & [e_2,e_3]=e_7,\ [e_3,e_5]=e_8,\ [e_4,e_6]=e_7 \\
 \widetilde{\mL} \quad & [e_1,e_3]=e_5,\ [e_1,e_4]=e_8,\
[e_1,e_6]=e_4,\
[e_2,e_3]=e_7,\\
            & [e_2,e_6]=e_8,\ [e_3,e_5]=e_8,\ [e_4,e_6]=e_7 \\
\end{array}
$$ These algebras are both indecomposable and nilpotent. They are
both the result of a contraction of $sl (3, \CC)$ and form so called
{\it continuous graded contractions} corresponding to the Pauli
grading of $sl (3, \CC)$\cite{PZ1,HN2}. They appear on the list in
\cite{HN2} named as $\mL_{17,9}$ and $\mL_{17,12}$. Computing their
invariants we obtain

$$
\begin{array}{llll}
\inv(\mL) \qquad & (8,4,0) (8,4,2,0)(2,5,8) & 2 & [16,19,9,11,8,17] \\
\inv(\widetilde{\mL})  & (8,4,0)(8,4,2,0)(2,5,8) & 2 & [16,19,9,11,8,17]. \\
\end{array}
$$ Here we observe that a unique characterization is still not
attained. On the contrary, computing invariant functions $\g,\f$
and $\widetilde{\g},\widetilde{\f}$ for algebras $\mL$ and
$\widetilde{\mL}$ yield:

\begin{center}
\begin{tabular}[t]{|l||c|c|c|}
\hline \parbox[l][20pt][c]{0pt}{}   $\alpha$ & 0 & -2 & \\ \hline
\parbox[l][25pt][c]{0pt}{} $\f(\alpha)$ & 19 & 17 & 16\\ \hline
\end{tabular}\qquad
\begin{tabular}[t]{|l||c|c|c|}
\hline  \parbox[l][20pt][c]{0pt}{}  $\alpha$ & 0 & 1 & \\ \hline
\parbox[l][20pt][c]{0pt}{} $\g(\alpha)$ & 16  & 9 & 8\\ \hline
\parbox[l][25pt][c]{0pt}{} $\widetilde{\g}(\alpha)$ & 16  & 9 & 8\\
\hline
\end{tabular}\qquad
\begin{tabular}[t]{|l||c|c|c|}
\hline \parbox[l][20pt][c]{0pt}{}   $\alpha$             & 0  &
$-\frac{1}{2}$ & \\ \hline \parbox[l][25pt][c]{0pt}{}
$\widetilde{\f}(\alpha)$ & 19 & 17            & 16\\ \hline
\end{tabular}
\end{center}
Since $\f \neq \widetilde{\f}$, we conclude that $\mL \ncong
\widetilde{\mL} $.
\end{example}

\begin{example}
We present Lie brackets for indecomposable eight-dimensional
nilpotent one--parametric Lie algebra: $$
\begin{array}{ll}
\mL(a) \quad & [e_1,e_3]=e_5,\ [e_1,e_4]=-ae_8,\ [e_2,e_3]=e_7,\
[e_2,e_4]=e_6,\\
            & [e_3,e_5]=e_8,\ [e_3,e_7]=e_6,\ 0 \leq |a| \leq 1 \\

\end{array}
$$ This continuum appeared as $\mL_{18,25}(a)$  in
\cite{HN2}, however, the relations among its algebras remained
unresolved there. We achieve partial characterization by isolating
two of its points, $a=0,-1$, and thus obtain $$
\begin{array}{llll}
 \inv(\mL(0)) \qquad & (8,4,0)
(8,4,2,0)(2,5,8) & 4 & [21,23,10,14,9,18] \\ \inv(\mL(-1)) \qquad
& (8,4,0) (8,4,2,0)(2,5,8) & 4 & [22,22,10,13,9,18] \\
\inv(\mL(a))\quad a\neq 0,-1 \qquad & (8,4,0) (8,4,2,0)(2,5,8) & 4
& [20,22,10,13,9,18] \\
\end{array}
$$ We summarize the tables of invariant functions $\g_a$ and $\f_a$ of
$\mL(a)$ as follows:
\begin{center}
\begin{tabular}[t]{|l||c|c|c|}
\hline \parbox[l][20pt][c]{0pt}{}   $\alpha$  & 0  & 1 & \\ \hline
\parbox[l][20pt][c]{0pt}{} $\g_a(\alpha)$ & 16 & 10 & 9\\ \hline
\end{tabular}
\vspace{6pt}

\begin{tabular}[t]{|l||c|c|}
\multicolumn{3}{l}{$a = 0$} \\ \hline \parbox[l][20pt][c]{0pt}{}
$\alpha$  & 0  &  \\ \hline \parbox[l][20pt][c]{0pt}{}
$\f_0(\alpha)$ & 23 & 21 \\ \hline
\end{tabular}\quad
\begin{tabular}[t]{|l||c|c|c|c|}
\multicolumn{5}{l}{$a=1$} \\ \hline \parbox[l][20pt][c]{0pt}{}
$\alpha$  & 0  & -1 & 1 &\\ \hline \parbox[l][20pt][c]{0pt}{}
$\f_1(\alpha)$ & 22 & 21 & 20 & 19\\ \hline
\end{tabular}\quad
\begin{tabular}[t]{|l||c|c|c|}
\multicolumn{4}{l}{$a=-1$} \\ \hline \parbox[l][20pt][c]{0pt}{}
$\alpha$  & 0  & 1 & \\ \hline \parbox[l][20pt][c]{0pt}{}
$\f_{-1}(\alpha)$ & 22 & 22 & 19\\ \hline
\end{tabular}

\vspace{6pt}

\begin{tabular}[t]{|l||c|c|c|c|c|}
\multicolumn{6}{l}{$a \neq 0, \pm 1$} \\ \hline
\parbox[l][20pt][c]{0pt}{}   $\alpha$  & 0  & 1 & $-a$ &
$-\frac{1}{a}$ &\\ \hline \parbox[l][20pt][c]{0pt}{}
$\f_a(\alpha)$ & 22 & 20 & 20 & 20 & 19\\ \hline
\end{tabular}
\end{center}
Finally, the invariant function $\f_a$ provides us with complete
characterization of the continuum; or more precisely: since for $a,
b \neq 0,\pm1, \ a\neq b,\ a \neq \frac{1}{b}$ it holds
 $$\f_a(-a)= 20 \neq \f_b(-a) = 19,$$
the relation $\mL(a) \ncong \mL(b)$ is thus guaranteed.
\end{example}

\begin{example}Similarly to the previous example, we list commutators of
indecomposable eight-dimensional nilpotent one--parametric Lie
algebra: $$
\begin{array}{ll}
\mL(a) \quad & [e_1,e_3]=ae_5,\ [e_1,e_4]=e_8,\ [e_1,e_6]=e_4,\
[e_2,e_3]=e_7,\\
  & [e_3,e_5]=e_8,\ [e_3,e_6]=e_2,\ [e_4,e_6]=e_7,\ 0 < |a| \leq 1 \\
\end{array}
$$ This class of Lie algebras also appeared named as
$\mL_{17,13}(a)$ in \cite{HN2} and the relations among its algebras
remained also unresolved. Isolating one of its points, $a=-1$ we
obtain $$
\begin{array}{llll}
\inv(\mL(-1)) \qquad & (8,5,0) (8,5,2,0)(2,5,8) & 4 &
[19,19,8,9,7,18] \\ \inv(\mL(a))\,\,\,\,a\neq -1 \qquad & (8,5,0) (8,5,2,0)(2,5,8)
& 4 & [17,19,8,9,7,18] \\
\end{array}
$$
Listing the invariant functions $\g_a$ and $\f_a$ for $\mL(a)$

\begin{center}
\begin{tabular}[t]{|l||c|c|c|}
\hline \parbox[l][20pt][c]{0pt}{}   $\alpha$  & 0  & 1 & \\ \hline
\parbox[l][20pt][c]{0pt}{} $\g_a(\alpha)$ & 16 & 8 & 7\\ \hline
\end{tabular}

\vspace{6pt}

\begin{tabular}[t]{|l||c|c|c|c|}
\multicolumn{5}{l}{$a=1$} \\ \hline \parbox[l][20pt][c]{0pt}{}
$\alpha$  & -1  & 0 & 1 &\\ \hline \parbox[l][20pt][c]{0pt}{}
$\f_1(\alpha)$ & 19 & 19 & 17 & 16\\ \hline
\end{tabular}\qquad
\begin{tabular}[t]{|l||c|c|c|c|}
\multicolumn{5}{l}{$a=-1$} \\ \hline \parbox[l][20pt][c]{0pt}{}
$\alpha$  & 0 & 1 & -1 &\\ \hline \parbox[l][20pt][c]{0pt}{}
$\f_{-1}(\alpha)$ & 19 & 19 & 17 & 16\\ \hline
\end{tabular}

\vspace{18pt}

\begin{tabular}[t]{|l||c|c|c|c|c|c|}
\multicolumn{7}{l}{$a \neq \pm1$} \\ \hline
\parbox[l][20pt][c]{0pt}{}   $\alpha$  & 0 & 1 & $-1$ & $-a$ &
$-\frac{1}{a}$ & \\ \hline \parbox[l][20pt][c]{0pt}{}
$\f_a(\alpha)$ & 19 & 17 & 17 & 17 & 17 & 16\\ \hline
\end{tabular}
\end{center}
enables us to conclude that the function $\f_a$ again represents a
priceless instrument providing a complete description of presented
parametric continuum of Lie algebras.

\end{example}

\section{Concluding remarks}
\begin{itemize}
    \item There are several non-equivalent ways of generalizing the notion of derivation of Lie algebra.
     For example in \cite{Hartwig}, a linear operator $A \in \End(\mL)$ is
    called a $(\sigma,\tau)$--derivation of $\mL$ if
    for some $\sigma, \tau \in \End(\mL)$ and all $x,y \in \mL$ $$ A[x,y] = [Ax,\tau y] +
    [\sigma x,Ay].$$ This generalization for $\sigma,\tau$
    homomorphisms appears already in \cite{Jacobson}. If there exists $ B\in \der(\mL)$ such that for
    all $x,y \in \mL$ the condition $ A[x,y] = [Ax,y] +
    [x,By]$ holds, then the operator $A$ forms another generalization
    \cite{Bresar}. More general definition emerged in
    \cite{Leger} and runs as follows: $A \in \End(\mL)$ is called {\it generalized derivation} of $\mL$
    if there exist $B,C\in \End(\mL)$ such that for all $x,y \in \mL$ the property $ C[x,y] = [Ax,y] +
    [x,By]$ holds.
    \item The sets $\mD(1,1,0)$ and $\mD(0,1,-1)$ are called
    {\it centroid} and {\it quasicentroid} respectively in
    \cite{Leger}. The inquiry under which conditions these sets
    coincide has also been discussed.
    \item Jordan algebra $\mD(1,1,-1)$ together with Lie algebras
    $\mD(0,1,1)$, $
    \mD(1,0,0) \cap \mD(0,1,1)$ still deserve
    further study concerning their structure and mutual relations.
    \item  Having a matrix $A=(A_{ij})$ and structure
    constants $c^k_{ij}$ of $\mL$ in some basis, then in order to
    $A\in\mD(\alpha,\beta,\gamma)$ one has to solve the system of
    linear homogeneous equations $$\sum_m( \alpha c_{ij}^m A_{mk} +
    \beta c_{mj}^k A_{mi} + \gamma c_{im}^k A_{mj} ) = 0, \quad
    i,j,k=1,2,\ldots,\dim \mL.$$
    This shows how computation of spaces
    $\mD(\alpha,\beta,\gamma)$ (in fact, due to Theorem~\ref{klass} one parameter is sufficient) and consequently the function $\f$ is comparatively very easy.
    This computation is also viable in higher dimensions --- as presented in Sect. \ref{higher}.
    \item Compared to the extensive usefulness of
    the function $\f$, we haven't found much use for the invariant function
    $\g$; the form of the function $\g$ is, however, non--trivial
    and its general behaviour represents an open problem.
    \item Notion of $(\alpha,\beta,\gamma)$--derivations and a theorem similar to Theorem~\ref{klass} can be derived for general commutative or
    anti--commutative algebra.
    \item  As expected, though function $\f$ forms a complete invariant in dimension three, this is no
    longer true in dimension four (but it still works nicely there), let alone in higher dimensions.
    It is likely that some characteristics similar to presented
    invariant functions, more general perhaps, could complete the scenery of
    invariants. One can only encourage such pursuit as it seems to be the right way to go.

\end{itemize}

\section*{Acknowledgements}
The authors are grateful to J.~Tolar for numerous stimulating and searching discussions. Partial support by the Ministry of Education of Czech Republic (projects MSM6840770039 and LC06002) is gratefully acknowledged. 


\end{document}